\documentclass[11pt]{article}
\usepackage{amssymb,latexsym,amsmath,amstext,amsthm} 
\usepackage{fullpage,color}
\newcommand{\digest}{\ensuremath{\sigma}}
\newcommand{\ID}{\ensuremath{\mathrm{ID}}}
\newcommand{\W}{\ensuremath{\mathcal{W}}}

\newcommand{\iv}{\ensuremath{\mathit{iv}}}
\newcommand{\iu}{\ensuremath{\mathit{iu}}}
\newcommand{\kv}{\ensuremath{\mathit{kv}}}
\newcommand{\ku}{\ensuremath{\mathit{ku}}}

\newcommand{\MMO}{\ensuremath{\mathsf{MMO}}}
\newcommand{\MMI}{\ensuremath{\mathsf{MMI}}}

\newcommand{\LNCS}{\emph{Lecture Notes in Computer Science}}

\newtheorem{theorem}{Theorem}[section]
\newtheorem{lemma}[theorem]{Lemma}
\newtheorem{corollary}[theorem]{Corollary}

\theoremstyle{definition}
\newtheorem{definition}{Definition}[section]
\newtheorem{example}{Example}[section]

\newtheorem{remark}{Remark}[section]

\begin{document}

\title{Optimal constructions for ID-based one-way-function key predistribution schemes
realizing specified communication graphs}
\author{Maura~B.~Paterson\\
Department of Economics,
Mathematics and Statistics\\ Birkbeck, University of London,
Malet Street, London WC1E 7HX, UK
\and
Douglas~R.~Stinson\thanks{D.~Stinson's research is supported by NSERC discovery grant 203114-11}
\\David R.\ Cheriton School of Computer Science\\ University of Waterloo,
Waterloo, Ontario, N2L 3G1, Canada
}
\date{\today}

\maketitle

\begin{abstract} We study a method for key predistribution in a network of $n$ users where pairwise keys are computed by hashing users' IDs along with secret information that has been (pre)distributed to the network users by a trusted entity. A  communication graph $G$ can be specified to indicate which pairs of users should be able to compute keys. We determine necessary and sufficient conditions for schemes of this type to be secure. We also consider the problem of minimizing the storage requirements of such a scheme; we are interested in the total storage as well as the maximum storage required by any user. Minimizing the total storage is NP-hard, whereas minimizing the maximum storage required by a user can be computed in polynomial time.
\end{abstract}

\section{Introduction}

Suppose we have a network of $n$ users and we want every pair of users to have
a secure {\it pairwise key}. When keying information is distributed ``ahead of time'' by a trusted
authority, we have a {\it key predistribution scheme}, or {\it KPS}. 
The {\it trivial} KPS assigns $n-1$ distinct pairwise keys to each user, so we say that
each user has {\it storage} equal to $n-1$. The {\it total storage} in the trivial scheme is $n(n-1)$,
since each of the $n$ users has storage equal to $n-1$. Assuming the pairwise keys are
chosen independently and uniformly at random from a specified keyspace, the trivial scheme is secure against
(maximum size) coalitions of size $n-2$, since each pairwise key is known only to two participants and cannot be computed even if all the remaining $n-2$ participants collude.

An interesting way to reduce the storage requirement, as compared to the trivial scheme, is to use
a {\it Blom Scheme} \cite{Bl83}. The Blom Scheme incorporates a {\it security parameter} denoted by $k$;
pairwise keys are unconditionally secure against coalitions of up to $k$  users.
Each user's storage requirement in the Blom Scheme is $k+1$ and this storage requirement is shown to 
be optimal by Blundo {\it et al.}\ \cite{BDHKVY}.

Another approach, due to Lee and Stinson \cite{LS05}, is called 
an {\it ID-based one-way-function key predistribution scheme}, also known as {\it IOS}.
Here, pairwise keys are computed by hashing public information along with secret information.
This technique can be used to construct a secure scheme in which every pair
of $n$ users has a pairwise key. Here we obtain security against maximum size coalitions,
but the security depends on the hash function used to construct the keys.
We will assume that the hash function can be modelled as a 
random function, i.e, the security analysis will be done in the standard {\it random oracle model} \cite{BR93}. 
The storage requirement of this scheme is reduced by almost half
as compared to the trivial scheme. In \cite{CAG13}, a very similar scheme was described that has
the same storage requirement as IOS. 

It is also possible to consider a more general setting where we specify a 
{\it communication graph} $G$ consisting of all the pairs
of users who we want to be able to share a secret key. A trivial KPS for a given communication graph $G$
would assign $\mathsf{deg}(u)$ keys to vertex $u$, for every vertex $u$ in $G$, where 
$\mathsf{deg}(u)$ denotes the degree of vertex $u$.
In the case of unconditionally secure schemes secure against coalitions of size $k$,
Blundo {\it et al.}\ \cite{BDHKVY} gave a construction where a vertex $u$ has storage requirement
$\min \{ k+1, \mathsf{deg}(u)\}$, for every vertex $u$.  Lee and Stinson's construction \cite{LS05}
applies to regular graphs of even degree $d$; the storage requirement for each vertex is
$1 + d/2$.

\subsection{Our contributions}

In Section \ref{graphs.sec}, we introduce some graph-theoretic terminology and basic results that are
required later in the paper.
In Section \ref{model.sec}, we define a general model for IOS for a specified communication graph $G$. 
In Section \ref{secure.sec}, 
we determine necessary and sufficient
conditions for an IOS to be secure. Our characterization of secure IOS involves certain {\it key graphs}.
We then use our characterization of secure IOS to investigate the storage requirements of the
these types of schemes. We show in Section \ref{star.sec} 
that it is sufficient to restrict our attention to schemes obtained
by decomposing the edges of the communication graph into ``stars''. In Section \ref{totalstorage.sec}, we investigate how to minimize the total storage in the scheme. In general, this turns out to be equivalent to determining
the size of a maximum independent set of vertices in $G$, which is an NP-hard problem. For complete
graphs, however, it is easy to give an exact answer.
In Sections \ref{maxstorage.sec} and \ref{MMO.sec}, we turn to the problem of minimizing the maximum storage required 
by any vertex of $G$. We give a complete solution for regular graphs. For arbitrary graphs, we can show that 
the optimal maximum storage can be determined in polynomial time by exploiting a connection
with the {\it  minimum maximum indegree} problem. 
In Section \ref{discuss.sec}, we discuss our results in comparison 
to the recent paper by Choi, Acharya and Gouda \cite{CAG13}.

\subsection{Graph-theoretic terminology}
\label{graphs.sec}

\begin{definition}
A {\em graph} is a pair $G = (V,E)$ where $V$ is a finite set  of {\em vertices}
and $E$ is
a set of {\em edges}, where every edge is a set of two vertices. 
We say that the edge $e = \{u,v\}$ {\em joins} the vertices $u$ and $v$, and
$u$ and $v$ are {\em incident with} $e$. We may also write an edge $\{u,v\}$
as $uv$. 
If $E$ is instead a multiset, then we say
that $G = (V,E)$ is a {\em multigraph}.
\end{definition}

\begin{definition}For a vertex $u$ in a graph $G$,
The {\em degree} of  $u \in V$, denoted $\mathsf{deg}(u)$, 
is the number of edges that are incident with $v$.
A graph is {\em $m$-regular} if every vertex has degree equal to $m$.
\end{definition}

For a graph $G = (V,E)$, we will denote the number of vertices (i.e., $|V|$) by $n$ and the number of 
edges (i.e., $|E|$) by $\epsilon$.

\begin{definition} Suppose $G = (V,E)$ is a graph and $u,v \in V$, $u \neq v$.
A {\em $(u,v)$-path} is a sequence of the form $v_1,e_1,v_2,e_2, \dots, e_{\ell},v_{\ell+1}$,
where $v_1, \dots , v_{\ell+1} \in V$ are distinct vertices, $e_1, \dots , e_{\ell} \in E$, $v_1 = u$,
$v_{\ell+1} = v$, and $e_i$ is incident with $v_i$ and $v_{i+1}$ for $1 \leq i \leq \ell$.
\end{definition}

\begin{definition} A graph is {\em connected} if, for
all distinct vertices $u$ and $v$, there exists a $(u,v)$-path.
\end{definition}

%\begin{remark}
%All graphs in the paper are assumed to be connected unless explicitly stated otherwise.
%\end{remark}

\begin{definition}
A {\em directed graph} is a pair $G = (V,E)$ where $V$ is a finite set  of {\em vertices}
and $E$ is
a set of {\em directed edges}, where every edge is an ordered pair of two distinct vertices.
The {\em indegree} of a vertex $v \in V$ is the number of directed edges $(u,v) \in E$.
The {\em outdegree} of a vertex $v \in V$ is the number of directed edges $(v,u) \in E$.
\end{definition}

\begin{definition}
A connected graph  (or multigraph) $G = (V,E)$ has an {\em eulerian circuit} if there exists a circuit
in which every edge is used exactly once. It is well-known that a connected (multi)graph has 
an eulerian circuit if and only if every vertex has even degree.
\end{definition}

\begin{definition}
A {\em complete graph} $K_n$ is a graph on $n$ vertices where every pair of vertices
are joined by an edge.
\end{definition}

\begin{definition}
A {\em complete bipartite graph} $K_{n_1,n_2}$ is a graph on $n_1+n_2$ vertices $V = V_1 \cup V_2$,
where $|V_1| = n_1$, $|V_2| = n_2$, and $E = \{ \{v_1,v_2\} : v_1 \in V_1, v_2 \in V_2 \}$.
\end{definition}

\begin{definition}
\label{star.def}
A {\em star} is a complete bipartite graph
$K_{1,m}$ for some $m \geq 1$. If $m > 1$, then the {\em centre} of the 
graph is the (unique) vertex of degree $m$. If $m=1$, then the graph is a single edge
and we can specify either endpoint to be the centre of the star.
Vertices in a star that are not the centre will be termed
{\em leaves}.
A {\em directed star} is obtained from a star by directing every edge from the
leaf to the centre.
\end{definition}

\begin{definition}
Given a graph $G = (V,E)$, an {\em independent set} is a subset of vertices
$V_0 \subseteq V$ such that $\{u,v\} \not\in E$ for every $u,v \in V_0$.
The size of a maximum independent set of vertices in a graph $G$ is denoted by $\alpha(G)$.
\end{definition}

\begin{theorem}[Rosenfeld \cite{Ro64}]
\label{rosenfeld.thm}
In any $m$-regular graph $G = (V,E)$, we have  \[ \alpha(G) \leq \min \{\lfloor n/2 \rfloor, n-m \},\]
where $n = |V|$.
\end{theorem}

\begin{definition}
Given a graph $G = (V,E)$, the {\em complement} of $G$ is the graph $\overline{G} = (V, \overline{E})$, 
where $\{u,v\} \in \overline{E}$ if and only if $\{u,v\} \not\in E$, for every $u,v \in V$, $u \neq v$.
\end{definition}

\begin{definition}
\label{join.def} Let $G = (V,E)$ and $H = (W,F)$ be graphs, where $V \cap W = \emptyset$.
The {\em join} of $G$ and $H$, is the graph denoted as $G \vee H$. The vertex set of
$G \vee H$ is $V \cup W$. The edges of $G \vee H$ consist of all edges in $E$ or $F$, 
and all $\{v,w\}$ with $v \in V$ and $w \in W$.
\end{definition}

\begin{definition}
\label{induced.def} Let $G = (V,E)$ be a graph and let $W \subseteq V$.
The {\em subgraph  of $G$ induced by $W$} is the graph $H = (W,F)$, 
where $F = \{ \{u,v\} \in E : u,v \in W\}$.
\end{definition}

\section{A general model for ID-based one-way-function KPS}
\label{model.sec}

In an ID-based one-way-function key predistribution scheme (or IOS), 
pairwise keys are constructed by hashing public information along with secret information.
We  describe a general model for IOS for
a given communication graph $G = (V,E)$ on $n$ vertices. We identify the $n$ vertices
in $V$ with a set of $n$ users, say $\mathcal{U}$.
Our first goal is to obtain a KPS in which two users
$u,v \in \mathcal{U}$ have a pairwise key $L_{u,v}$ whenever$\{u,v\} \in E$.
Security of the scheme will be addressed in Section \ref{secure.sec}.
%and show that, in order to have a 
%secure KPS, we are required to follow the technique described above.

Let $f \colon E(G) \rightarrow \{1, \dots , t\}$ be a publicly known surjective function and let
$R_1, \dots , R_t \in \{0,1\}^{\digest }$ be secret values chosen uniformly at random from $\{0,1\}^{\digest }$.
Let $h \colon \{0,1\}^* \rightarrow \{0,1\}^{\digest }$ be a hash function (which we model as a random oracle). 
The length of the $R_i$'s, namely $\digest$, 
is the same as the length of the output of $h$.
A suitable value for $\digest$ is $\digest = 128$.
The function $h$ will be used as a {\it key derivation function}.

Every  pairwise key $L_{u,v}$ will be  computed as 
\begin{equation}
\label{key.eq}
L_{u,v} = L_{v,u} = h(R_{f(\{u,v\})} \parallel \ID(u) \parallel \ID(v)),
\end{equation}
where 
\begin{description}
\item[(1)] for every $u \in \mathcal{U}$, $\ID(u)$ denotes public identifying information for user $u$, and
\item[(2)] $u < v$ (this requirement ensures
that $L_{u,v} = L_{v,u}$ in (\ref{key.eq})). 
\end{description} % We say that $(f,\W)$ is a \emph{$G$-IOS}.

Suppose a value $R_{f(\{u,v\})}$ is not known to a coalition of users. Recall that $R_{f(\{u,v\})}$
was chosen uniformly at random from a set of size $2^{\digest}$. Also, the key $L_{u,v}$ is the output of
a random function $h$ that takes on values from a set of size $2^{\digest}$. Therefore this key
can be regarded as a secure $\digest$-bit key.

Eventually, we will investigate whether secret values may be ``repeated'', i.e., used for the
computation of more than one key, without compromising security. The motivation is
that this might enable the storage of the scheme to be reduced. This question will 
be addressed in Section \ref{secure.sec}.

\begin{remark}
As described above, the public inputs to the
hash function are taken to be (public) IDs. Actually, it doesn't really matter what the
public inputs to $h$ are, as long as there do not exist two keys that have the same
public and secret inputs and we ensure that $L_{u,v} = L_{v,u}$.
\end{remark}

It is obvious that if $\{u,v\} \in E$ and $f(\{u,v\}) = j$, then 
users $u$ and $v$ each must store 
one of $R_{j}$ or $L_{u,v}$. 
We will assume that neither $u$ nor $v$ stores both
$R_{j}$ and $L_{u,v}$, since $L_{u,v}$ can be computed if $R_j$ is known.

\begin{definition}
For $1 \leq j \leq t$, define 
\[E_j = f^{-1}(j) = \{ \{u,v\} \in E: f(\{u,v\}) = j\}\] and let $V_j$ be the vertices
spanned by $E_j$.
$(V_j,E_j)$ is a graph that we term the {\em $j$th key graph}.

%Clearly, for any edge $\{u,v\} \in E_j$,  users $u$ and $v$ each must store 
%one of $R_{j}$ or $L_{u,v}$. 
Let \[W_j = \{ v \in V_j : v \mbox{ stores } R_j\}\]
and define  $\W = (W_1, \dots , W_t)$. 
Since we are assuming that $f$ is surjective, it follows that
$|E_j| \geq 1$ and $|V_j| \geq 2$ for all $j \in \{1, \dots , t\}$.
%and let
%\[W'_j = V_j \backslash W_j.\]

An IOS is fully specified by 
$f$ and $\W$, so we will refer to  $(f,\W)$ as a \emph{$G$-IOS}.
\end{definition}

To summarize, here is all the keying information that is stored in a $G$-IOS, $(f,\W)$.
Suppose $u \in V$ is any vertex and  $\{u,v\} \in E$ is any edge that is incident with $u$.
If $\{u,v\} \in E_j$, 
then the vertex $u$ stores
$R_j$ when $u \in W_j$, and it stores $L_{u,v}$, otherwise.

\begin{example} 
\label{exam1}
Consider the graph $G = (V,E)$, where
\begin{center}$V = \{1,2,3,4,5\}$ and $E = \{12,13,14,15,23,25,34,45\}$.\end{center}
Suppose we define 
\begin{align*}
&f(12) = f(13) = f(14) = 1 &&W_1 = \{1\} \\
&f(15) = f(25) = f(45) = 2 &&W_2 = \{5\}\\
&f(23) = f(34) = 3 &&W_3 = \{3\}.
\end{align*}
Then it is easy to see that
\begin{align*}
&E_1 = \{12,13,14\} &&V_1 = \{1,2,3,4\} \\
&E_2 = \{15,25,45\} &&V_2 = \{1,2,4,5\}\\
&E_3 = \{23,34\} &&V_3 = \{2,3,4\}.
\end{align*}
The keying material held by each user is as follows:
\begin{center}
\begin{tabular}{l}
user 1 stores $R_1$, $L_{1,5}$\\
user 2 stores $L_{1,2}$, $L_{2,3}$, $L_{2,5}$\\
user 3 stores $R_3$, $L_{1,3}$\\
user 4 stores $L_{1,4}$, $L_{3,4}$, $L_{4,5}$\\
user 5 stores $R_2$.
\end{tabular}
\end{center}
$\blacksquare$
\end{example}

\subsection{Secure IOS}
\label{secure.sec}

Suppose that $(f,\W)$ is a $G$-IOS. We say that $(f,\W)$ is \emph{secure}
if there does not exist a user $w$ who can compute a key $L_{u,v}$ where $\{u,v\} \in E$ and
$w \neq u,v$. This security condition will be satisfied
provided that $w$ does not know either of the values $R_{j}$ or $L_{u,v}$
(where $j = f(\{u,v\})$). Note that a secure IOS
is automatically secure against maximum size coalitions.

We now consider when a $G$-IOS, as defined above, will be secure.
We prove a sequence of simple lemmas that culminate in a characterization of
secure $G$-IOS. The first lemma is obvious.

\begin{lemma}
\label{L0}
Suppose $(f,\W)$ is a $G$-IOS, $\{u,v\} \in E$ and $w$ stores $L_{u,v}$, where $w \neq u,v$. Then $(f,\W)$ 
is not secure.
\end{lemma}

%Henceforth, we will assume that there is no vertex $w$ that stores $L_{u,v}$, where $w \neq u,v$.

\begin{lemma}
\label{L-1}
Suppose $(f,\W)$ is a secure $G$-IOS and $w$ stores $R_j$. Then $w \in V_j$.
\end{lemma}
\begin{proof}
Suppose $w$ stores $R_j$ but $w \not\in V_j$. Let $\{u,v\} \in E_j$; then
$w \neq u,v$ because $w \not\in V_j$. But then $w$ can compute $L_{u,v}$, so 
$(f,\W)$ is not secure.
\end{proof}

\begin{lemma}
\label{L1}
Suppose $(f,\W)$ is a $G$-IOS and $\{u,v\}, \{u,v'\} \in E_j$, where $v \neq v'$. 
If $v \in W_j$  or $v' \in W_j$, then $(f,\W)$ 
is not secure.
\end{lemma}
\begin{proof}
If $v \in W_j$, then $v$ can compute $L_{u,v'}$.
Similarly, if $v' \in W_j$, then $v'$ can compute $L_{u,v}$.
\end{proof}

\begin{lemma}
\label{L2}
Suppose $(f,\W)$ is a $G$-IOS and $\{u,v\}, \{u',v'\} \in E_j$, 
where $\{u,v\} \cap \{u',v'\} = \emptyset$. 
If $\{u,u',v,v'\} \cap W_j \neq \emptyset$, then $(f,\W)$ 
is not secure.
\end{lemma}
\begin{proof}
If $u \in W_j$, then $u$ can compute $L_{u',v'}$.
If $v \in W_j$, then $v$ can compute $L_{u',v'}$.
If $u' \in W_j$, then $u'$ can compute $L_{u,v}$.
If $v' \in W_j$, then $v'$ can compute $L_{u,v}$.
\end{proof}

\begin{lemma}
\label{L3}
Suppose $(f,\W)$ is a secure $G$-IOS and $u,v \in W_j$. Then 
the key graph $E_j$ consists of a single edge $\{\{u,v\}\}$. % and hence $|E_j| = 1$.
\end{lemma}
\begin{proof}
First suppose that $\{u,v\} \not\in E_j$. Suppose
$\{u,u'\} \in E_j$ and $\{v,v'\} \in E_j$. If $u' \neq v'$, then Lemma \ref{L2}
is contradicted. If $u' = v'$, then Lemma \ref{L1}
is contradicted. Therefore it follows that  $\{u,v\} \in E_j$.
Suppose that there is some other edge $\{u',v'\} \in E_j$. If
$\{u,v\} \cap \{u',v'\} = \emptyset$, then Lemma \ref{L2} is contradicted.
Otherwise, suppose without loss of generality that $u = u'$. Then Lemma \ref{L1}
is contradicted.
\end{proof}

\begin{lemma}
\label{L4}
Suppose $(f,\W)$ is a secure $G$-IOS. Then, for every $j$, $|W_j| = 0, 1$ or $2$.
\end{lemma}
\begin{proof}
Suppose $|W_j| \geq 3$ and suppose  $W_j$ contains three distinct
elements $u,v,w$. 
We now apply Lemma \ref{L3}: since $u,v \in W_j$, we have that 
$E_j = \{\{u,v\}\}$. Furthermore, since $u,w \in W_j$, we have that 
$E_j = \{\{u,w\}\}$. This is impossible because $v \neq w$.
\end{proof}

\begin{lemma}
\label{L5}
Suppose $(f,\W)$ is a secure $G$-IOS. 
If $|W_j| = 1$, say $W_j = \{v\}$, then the $j$th key graph  is a star with centre $v$.
\end{lemma}
\begin{proof}
Suppose there is an edge $\{u,u'\} \in E_j$ where $v \neq u,u'$. Then $v$ can compute $L_{u,u'}$.
\end{proof}

The above conditions are necessary for a $G$-IOS, say $(f,\W)$, to be secure; we now show that they 
are also sufficient to provide security.
%We can show that additional conditions can be assumed for storage-optimal KPS.

\begin{theorem}
\label{T1}
Suppose $(f,\W)$ is a  $G$-IOS. Then $(f,\W)$ is secure if and only if the following conditions hold:
\begin{description}
\item[(1)] $|W_j| \leq 2$ for all $j$; 
\item[(2)] for any $j$ with $|W_j| = 1$, say $W_j = \{v\}$, the $j$th key graph   is a star with centre $v$; and 
\item[(3)] for any $j$ with $|W_j| = 2$, say $W_j = \{u,v\}$,  the $j$th key graph  consists of the single edge
$\{\{u,v\}\}$.
\end{description}
\end{theorem}
\begin{proof}
These three conditions are shown to be necessary for $(f,\W)$ to be 
secure in Lemmas \ref{L3}--\ref{L5}. We now show they
are sufficient for the scheme to be secure.
Suppose the three conditions hold but $(f,\W)$ is not secure. We will obtain a contradiction.

If $(f,\W)$ is not secure, then some user $w \neq u,v$ can compute a key
$L_{u,v}$ where $\{u,v\} \in E$. Let $f(\{u,v\}) = j$. 
Since $w$ does not store $L_{u,v}$ (Lemma \ref{L0}), it follows 
from Lemma \ref{L-1} that $w \in W_j$ and hence $|W_j| \geq 1$.
Applying Lemma \ref{L4}, we have that $|W_j| = 1$ or $2$.

Suppose $|W_j| = 2$. Clearly $u,v,w \in V_j$ and hence $|V_j| \geq 3$.
This contradicts Lemma \ref{L3} which says that $E_j$ consists of a single edge.
 
Suppose $|W_j| = 1$ (so $W_j = \{w\}$). 
From Lemma \ref{L5}, it follows that $E_j$ is a star with centre $w$. But then the existence of
the edge $\{u,v\} \in E_j$ yields a contradiction.
\end{proof}

It may be helpful to summarize the  cases enumerated in Theorem \ref{T1} in more descriptive language.
We will say that the $j$th key graph  is of {\em type $i$} if $|W_j| = i$. Then Theorem \ref{T1} can be
restated by saying that every key graph is of type $0$, $1$ or $2$. Furthermore, the structure of 
a key graph of a specified type is as follows:
\begin{center}
\begin{tabular}{lp{5.5in}}
{\bf type 0} & In a key graph $(V_j,E_j)$ of type $0$, no vertex in $V_j$ stores the value $R_j$.
Hence, for every vertex $u \in V_j$ and every edge $\{u,v\} \in E_j$ that is incident with $u$,
the vertex $u$ stores the key $L_{u,v}$. There is no restriction on the number of edges in $E_j$ or the
structure of $E_j$.\vspace{.05in}\\
{\bf type 1} & A key graph $(V_j,E_j)$ of type $1$ is a star whose centre (say $u$) stores the value $R_j$. Any leaf $v \in V_j$ stores the key $L_{u,v}$.\vspace{.05in}\\
{\bf type 2} & A key graph $(V_j,E_j)$ of type $2$ consists of a single edge $\{u,v\}$ where $u$ and $u$ both 
store the value $R_j$.
\end{tabular}
\end{center}

\begin{remark}
The IOS considered in Example \ref{exam1} consists of three key graphs of type 1.
\end{remark}

\subsection{Edge-decompositions into stars}
\label{star.sec}

We say that a (secure) $G$-IOS $(f,\W)$ is a {\it star-IOS} if every key graph is of type $1$.
That is, the KPS is based on an edge-decomposition of $G$ into stars. This was the
model introduced by Lee and Stinson \cite{LS05}. In this section, we show that the 
any secure $G$-IOS can be transformed into a (secure) $G$-star-IOS in which the storage of
each vertex is the same in the two schemes.

Basically, we need to describe how to change type $0$ and $2$ key graphs into type
$1$ key graphs. Let's begin by considering  a type $2$ key graph, say
$(V_j,E_j)$, which consists of single edge $\{u,v\}$. Both vertices $u$ and $v$ store
$R_j$. If we stipulate that one of vertices $u$ or $v$ stores $R_j$ and the other one
stores $L_{u,v}$, then $E_j$ has been transformed to a type $1$ key graph with the
same storage requirements.

Now we suppose we have a type $0$ key graph, $(V_j,E_j)$. No vertex in $V_j$ stores the
value $R_j$. In this case, we can split this key graph into $|E_j|$ key graphs of type
$1$, each of which is isomorphic to $K_{1,1}$. Every edge is now assigned a {\it different}
random value by the function $f$. Furthermore, for every edge $e = \{u,v\} \in E_j$, 
one endpoint stores the (new)
random value $f(e)$ and the other endpoint stores the (new) key $L_{u,v}$.

Summarizing the above discussion, we have the following theorem. 

\begin{theorem}
\label{star.thm}
If $(f,\W)$ is a secure $G$-IOS, then there exists a secure $G$-star-KPS in which the storage of
every vertex is the same in both schemes.
\end{theorem}

\section{Optimal storage of IOS}
\label{storage.sec}

In this section, we focus on secure $G$-IOS that have minimum possible storage requirements.
In view of Theorem \ref{star.thm}, we can restrict our attention to $G$-star-IOS.
We are interested in the total storage required in such a scheme, as well as
the (maximum) storage required by an individual user.

\subsection{Total storage}
\label{totalstorage.sec}

Let  $G = (V,E)$ be a graph with $|V| = n$ and $|E| = \epsilon$.
Suppose that $(f,\W)$ is a secure $G$-star-IOS. Then the
 {\it storage requirement} of a node $u$, denoted $s(u)$, is defined to be
 the number of stars (i.e., key graphs)
that contain $u$. The {\it total storage requirement} of the IOS, denoted by $S(f,\W)$, is defined
as \[  S(f,\W) = \sum_{u \in V} s(u).\]
Let $c(u)$ denote the number of stars for which $u$ is the centre and
let $\ell(u)$ denote the number of stars for which $u$ is a leaf.
Then 
\begin{equation}
\label{s.eq}
s(u) = c(u) + \ell(u)
\end{equation} and hence
\[ S(f,\W) = \sum_{u \in V} (c(u) + \ell(u)).\]
If we define \[  C = \sum_{u \in V} c(u),\]
then we have that
\begin{equation}
\label{S.eq}
 S(f,\W) = C + \sum_{u \in V}\ell(u).
\end{equation}

\begin{example}
\label{exam2}
In Example \ref{exam1}, we have the following storage requirements:
\[
\begin{array}{c|c|c|c}
i & c(i) & \ell(i) & s(i) \\ \hline
1 & 1 & 1 & 2 \\
2 & 0 & 3 & 3 \\
3 & 1 & 1 & 2 \\
4 & 0 & 3 & 3 \\
5 & 1 & 0 & 1 
\end{array}
\]
The total storage of the scheme is $11$.
$\blacksquare$
\end{example}

%Given an edge-decomposition of $G$ into stars, there is a natural way to orient 
%the edges in $E$, namely, direct every edge form the leaf to the centre.
%We say that this orientation is the orientation \emph{associated with} the given edge-decomposition.

We now consider each star in the edge-decomposition to be directed as defined in Definition \ref{star.def}, i.e., each edge is directed from the leaf to the centre of the star containing it.

\begin{lemma}
\label{S.lem}
$\sum_{u \in V} \ell(u) = \epsilon$.
\end{lemma}

\begin{proof}
It is obvious that $\ell(u) = d^{+}(u)$ (i.e., the outdegree of $u$).
Clearly $\sum_{u \in V} d^{+}(u) = \epsilon$, and the result follows.
\end{proof}

It follows from (\ref{S.lem}) and Lemma \ref{S.eq} that 
\begin{equation}
\label{SC.eq}
S(f,\W) = \epsilon + C.
\end{equation}
Further, it is easy to see that
\begin{equation}
\label{c.eq}
\begin{array}{ll}
c(u)= 0 & \mbox{if $d^{-}(u) = 0$} \\
c(u) \geq 1 & \mbox{if $d^{-}(u) > 0$}.
\end{array}
\end{equation}

We are interested in minimizing the total storage, we will denote by
$S^*(G)$ the minimum value of $S(f,\W)$ over all secure $G$-star-IOS.
In order to compute $S^*(G)$, we need to minimize the value of $C$ in (\ref{SC.eq}).
It is convenient to let $C_{\mathrm{min}}(G)$ denote the minimum possible value of $C$ over
all edge-decompositions of $G$ into stars.

\begin{lemma}
Let $\alpha = \alpha(G)$ denote the size of a maximum independent set of vertices in $G$. 
Then $C_{\mathrm{min}}(G) = n - \alpha$. Furthermore, $C = C_{\mathrm{min}}(G)$ only if 
$c(u) \in \{0,1\}$ for all vertices $u$.
\end{lemma}

\begin{proof}
Let $V_0$ be a set of $\alpha$ independent vertices in $V$. Direct all edges
that are incident with a vertex $v \in V_0$ away from $v$, and direct any 
remaining edges arbitrarily. This shows that $C_{\mathrm{min}} \leq n - \alpha$.

Conversely, observe that the set of vertices for which $c(u) = 0$ form an independent set in $G$.
This yields the bound $C_{\mathrm{min}} \geq n - \alpha$. Further, in order for this bound to be met
with equality, $c(u) \leq 1$ for all vertices $u$.
\end{proof}

\begin{theorem}
\label{S*bound}
For any graph $G$, $S^*(G) =  n + \epsilon - \alpha$.
\end{theorem}

\begin{remark}
For the graph considered in Example \ref{exam1}, it is easy to see that $\alpha = 2$.
Theorem \ref{S*bound} then yields $S^*(G) = 11$. The scheme constructed in
Example \ref{exam1} has  total storage 11, as noted in Example \ref{exam2}. Therefore
this scheme has optimal total storage.
\end{remark}

We now consider the situation where $G$ is a complete graph $K_n$. 
The following result is an immediate corollary of Theorem  \ref{S*bound};
it improves the constructions given in \cite{LS05,CAG13} by one.

\begin{corollary}
$S^*(K_n) = \binom{n+1}{2} - 1$. Further, an edge-decomposition of $K_n$ into stars
meets this bound with equality if and only if there is a (necessarily unique)
vertex $v$ with $d^{-}(v) =  0$.
\end{corollary}

It is well-known that computing the exact value of $\alpha(G)$ is \textsf{NP}-hard.
Thus we also have the following corollary of Theorem  \ref{S*bound}.

\begin{corollary}
Given a graph $G$, computing $S^*(G)$ is \textsf{NP}-hard.
\end{corollary}

\subsection{Maximum storage}
\label{maxstorage.sec}

It is also of interest to consider the \emph{maximum storage} of a secure $G$-IOS, say $(f,\W)$,  
which is defined to be $S_{\mathrm{max}}(f,\W) = \max \{ s(u) : u \in V \}$.
Define $S_{\mathrm{max}}^*(G)$ to be the minimum value of $S_{\mathrm{max}}(f,\W)$ over all $(f,\W)$ that
are secure $G$-IOS. As before, we can restrict our attention to  $(f,\W)$ that
are secure $G$-star-IOS.

We begin with a lemma that states a simplification we can make without loss of generality.

\begin{lemma}
\label{c=1.lem}
For any graph $G$, there exists a secure $G$-IOS having optimal maximum storage $S_{\mathrm{max}}^*(G)$ 
in which $c(u) \in \{0,1\}$ for all vertices $u$.
\end{lemma}

\begin{proof} 
Consider the star-decomposition associated with a secure $G$-IOS having optimal maximum storage $S_{\mathrm{max}}^*(G)$. If $c(u) > 1$ for some vertex $u$, then merge all the stars having centre $u$.
This reduces the storage $s(u)$ and leaves the storage of all other vertices unchanged. Repeat this process
until $c(u) \in \{0,1\}$ for all vertices $u$.
\end{proof}

We next give a construction that yields an upper bound on $S^*(G)$.
This is a slight modification of a construction given by Lee and Stinson \cite{LS05} that 
only applied to regular graphs of even degree.

\begin{theorem}
\label{euler.thm}
Let $G = (V,E)$ be a graph and let $d$ denote the maximum degree of any vertex in $V$.
Then 
\[S_{\mathrm{max}}^*(G)  \leq
\begin{cases}
\frac{d+3}{2} & \text{if $d$ is odd}, \\
\frac{d+2}{2} & \text{if $d$ is even}.
\end{cases}
\]
\end{theorem}

\begin{proof}
We first assume that $G$ is connected.
Let $V_0 \subseteq V$ be the vertices in $V$ that have odd degree.
Clearly $|V_0|$ is even. Let $M$ be any matching of the vertices in $V_0$;
$M$ consists of $|V_0|/2$ disjoint edges. Now consider the multigraph
$G' = (V, E' = E \cup M)$. Every vertex in $G'$ has even degree, so $G'$ has a (directed)
eulerian circuit, say $D$.

The important property is that, with respect to this orientation defined on the edges in $E'$, 
$d^+(u) = d^-(u)$ for every vertex $u$. Now remove the edges in $M$ and consider the resulting 
orientation on the edges in $E$. We have $|d^{-}(v) - d^{+}(v)| \leq 1$ if $d(v)$ is odd and
$d^{-}(v) = d^{+}(v)$ if $d(v)$ is even.
This orientation gives rise to an associated 
decomposition of $E$ into stars, where each star consists of the edges directed into a vertex. 

To complete the proof, we observe that any vertex $v$ has storage 
\[s(v) = c(v) + \ell(v) = c(v) + d^{-}(v) \leq 1 + d^{-}(v).\]
If $d$ is even, then $d^{-}(v) \leq d/2$ for all $v$, and if $d$ is odd, then
$d^{-}(v) \leq (d+1)/2$ for all $v$.
The desired result follows.

If $G$ is not connected, then apply the above-described technique to every connected
component of $G$.
\end{proof}

If we start with a regular graph degree at least two, then we can show that the result obtained above is optimal.
 
\begin{theorem}
\label{regular-1.thm}
Let $G = (V,E)$ be a graph that is regular of even degree $d>0$.
Then $S_{\mathrm{max}}^*(G) = (d+2)/2$.
\end{theorem}

\begin{proof}
We have that $S^*(G) = n + \epsilon - \alpha$. Here $\epsilon = nd/2$ since $G$ is $d$-regular.
Therefore, 
\[ S_{\mathrm{max}}^*(G) \geq \left\lceil \frac{n + \epsilon - \alpha}{n}  \right\rceil 
= \left\lceil 1 + \frac{d}{2} - \frac{\alpha}{n}  \right\rceil .\]
Since $1 \leq \alpha < n$, it follows that
$S_{\mathrm{max}}^*(G) \geq 1 + d/2$. We $S_{\mathrm{max}}^*(G) \leq 1 + d/2$ from Theorem \ref{euler.thm},
so the result follows.
\end{proof}

For regular graphs of odd degree, we have the following similar result.

\begin{theorem}
Let $G = (V,E)$ be a graph that is regular of odd degree $d\geq 3$.
Then $S_{\mathrm{max}}^*(G) = (d+3)/2$.
\end{theorem}

\begin{proof}
As in the proof of the Theorem \ref{regular-1.thm}, 
\[ S^*(G) \geq \left\lceil \frac{n + \epsilon - \alpha}{n}  \right\rceil 
= \left\lceil 1 + \frac{d}{2} - \frac{\alpha}{n}  \right\rceil .\]
In any $m$-regular graph, we have  $\alpha \leq n/2$ from Theorem \ref{rosenfeld.thm}. 
Therefore  
\[ S_{\mathrm{max}}^*(G) \geq \left\lceil 1 + \frac{d}{2} - \frac{n/2}{n}  \right\rceil \geq 
\left\lceil  1 + \frac{d-1}{2}  \right\rceil = \frac{d+1}{2}.\] 
Now, in order for this bound to be met with equality, it must happen that
$C = \alpha = n/2$ and the total storage  $S^*(G) = n(d+1)/2$. 
In this case there will be $n/2$ vertices with $s(v) = 1$ and 
$n/2$ vertices with $s(v) = d$. So  $S_{\mathrm{max}}^*(G) = d$.
However, $(d+1)/2 < d$ when $d \geq 3$, so we have a contradiction.
Therefore,   $S_{\mathrm{max}}^*(G)  \geq (d+1)/2 + 1 = (d+3)/2$. 
$S_{\mathrm{max}}^*(G)  \leq (d+3)/2$ from Theorem \ref{euler.thm},
so the result follows.
\end{proof}

\begin{remark}
A $1$-regular graph $G$ is a union of disjoint edges. If is easy to see that 
$S_{\mathrm{max}}^*(G) = 1$ for such a graph $G$.
\end{remark}

\subsection{Minimizing maximum indegree and outdegree}
\label{MMO.sec}

The problem of computing $S_{\mathrm{max}}^*(G)$ is closely related to the
{\it minimum maximum indegree} problem \cite{FG76,Ve04}, which
is defined as follows.
Given a graph $G$, the goal is to direct the edges in $G$ so as to minimize the 
maximum indegree of a vertex of the resulting directed graph. Suppose we denote this 
quantity by $\MMI(G)$. It was shown in \cite{Ve04} that $\MMI(G)$ can be computed 
in polynomial time, more specifically in time $O(\epsilon^2)$, 
where $\epsilon$ is the number of edges in
the graph. An improved algorithm  can be found in \cite{AMOZ07}. 
These algorithms also find an orientation attaining the optimal
value $\MMI(G)$.

We will use the following simple result a bit later.

\begin{lemma}
\label{MMIbound.lem}
 $\MMI(G) \geq \lceil \frac{\epsilon}{n} \rceil$.
 \end{lemma}
 
 \begin{proof} Given any orientation of the edges of a graph $G$, the 
 {\it average} indegree is  $\epsilon/n$.
 \end{proof}

We can analogously define the {\it minimum maximum outdegree} problem 
and the associated quantity $\MMO(G)$ in the obvious way.
Observe that  $\MMO(G)  = \MMI(G)$, simply by reversing the directions
of all edges in an optimal solution.

Here is our result linking maximum storage of a $G$-IOS to $\MMO(G)$.

\begin{theorem}
\label{MMO.thm}
For any graph $G$, we have
$\MMO(G) \leq S_{\mathrm{max}}^*(G) \leq \MMO(G) + 1$.
\end{theorem}

\begin{proof}
Suppose we have with a star-decomposition of $G$ that minimizes $S_{\mathrm{max}}^*(G)$
and suppose that all edges are directed towards the centres of the stars in the decomposition,
as usual. From (\ref{s.eq}), the storage of a vertex $u$ is $s(u) = \ell(u) + c(u)$.
 The value $\ell(u)$ is clearly equal to the outdegree of $u$.
 From Lemma \ref{c=1.lem}, we can assume $c(0) = 0$ or $1$.
 The result follows.
\end{proof}

We expect for ``most'' graphs that $S_{\mathrm{max}}^*(G) = \MMO(G) + 1$. The only way it can occur
that
$\MMO(G) = S_{\mathrm{max}}^*(G)$ is if every vertex whose outdegree is equal to $\MMO(G)$
has indegree equal to $0$. Nevertheless, this can occur:
we show that there are infinitely many graphs for which 
$\MMO(G) = S_{\mathrm{max}}^*(G)$. 

\begin{theorem}
\label{construction}
For any integer $n > 15$, there exists a graph $G$
having $n$ vertices such that $\MMO(G) = S_{\mathrm{max}}^*(G)$.
\end{theorem}

\begin{proof}
For any integer $t > 10$, let $G = K_5 \vee \overline{K_t}$ 
(i.e., $G$ is the join of $K_5$ and a set of $t$ independent vertices; see Definition \ref{join.def}).
This graph has $n = t+5$ vertices and $\epsilon = 5t+10$ edges.  From Lemma \ref{MMIbound.lem},
since $t > 10$, 
we see that $\MMO(G) \geq 5$. We can construct an orientation of the edges of 
$G$ in which the $t$ vertices in the $\overline{K_t}$ each have outdegree 5 and 
indegree 0, and the vertices in the $K_5$ each have outdegree 2 and indegree $t+2$.
This shows that $\MMO(G) \leq 5$, whence $\MMO(G) = 5$. 
 Since the vertices of outdegree 5 all have indegree equal to 
$0$, the resulting star decomposition proves that $S_{\mathrm{max}}^*(G) = \MMO(G) = 5$.
\end{proof}

Since $\MMO(G)$ can be computed in polynomial time, Theorem \ref{MMO.thm} establishes that
we can compute an integer $T$ such that value $S_{\mathrm{max}}^*(G) = T$ or $T+1$, in 
polynomial time. In fact, as we now show, it is possible to compute the exact value of 
$S_{\mathrm{max}}^*(G)$ in polynomial time.

\begin{theorem}
\label{MMO2.thm}
Let $G = (V,E)$ be a graph and let 
\[W = \{ v \in V : \mathsf{deg}(v) >  \MMO(G) \}.\]
Let $H = (W,F)$ be the subgraph of $G$ induced by $W$. Then
$\MMO(G) = S_{\mathrm{max}}^*(G)$ if and only if $\MMO(H) < \MMO(G)$.
\end{theorem}

\begin{proof}
Suppose $\MMO(G) = S_{\mathrm{max}}^*(G)$ and consider an orientation of
the edges in $E$ so that the maximum outdegree of any vertex is $\MMO(G)$.
Let $w \in W$; then $w$ has indegree greater than $0$. It follows that 
$s(w) = 1 + \ell(w)$, where $\ell(w)$ equals the outdegree of $w$.
Therefore \[\ell(w) < s(w) \leq S_{\mathrm{max}}^*(G) = \MMO(G).\]
Since this holds for every vertex $w \in W$, we have $\MMO(H) < \MMO(G)$.

Conversely, suppose that $\MMO(H) < \MMO(G)$. Consider an orientation of
the edges in $F$ so that the maximum outdegree of any vertex in $W$ is at most $\MMO(H)$.
Next, direct any edges having one endpoint in
$W$ towards the incident vertex in $W$. 
At this point, we have 
\[s(w) \leq 1 + \ell(w) \leq 1 + \MMO(H) \leq \MMO(G)\]
for every vertex $w \in W$. Finally, direct any edges with both
endpoints in $V \setminus W$ arbitrarily. Then, for any vertex $v \in V \setminus W$,
we have \[s(v) \leq \mathsf{deg}(v) \leq \MMO(G).\]
Thus $\MMO(G) = S_{\mathrm{max}}^*(G)$.
\end{proof}

\begin{example}
Let us return again to the graph $G$ considered in Example \ref{exam1}.
It is easy to see that $\MMO(G) = 2$. Further, $H = G$ and hence $S_{\mathrm{max}}^*(G) = 3$.
Thus the scheme constructed in Example \ref{exam1} has optimal maximum storage.
\end{example}

\begin{example}
In the construction given in the proof of Theorem \ref{construction}, we see that $H$ consists of the 
five vertices in the $K_5$ and hence $\MMO(H) = 2$.
\end{example}

As an immediate corollary of Theorem \ref{MMO2.thm}, we obtain a polynomial-time algorithm to compute $S_{\mathrm{max}}^*(G)$. This algorithm is presented in Figure \ref{fig1}.

\begin{figure}[tb]
\caption{An algorithm to compute $S_{\mathrm{max}}^*(G)$}
\label{fig1}
\begin{center}
\framebox{
\begin{minipage}{6in}
\begin{enumerate}
\item Compute $\MMO(G)$.
\item Construct the graph $H$.
\item Compute $\MMO(H)$.
\item If $\MMO(H) < \MMO(G)$ then $S_{\mathrm{max}}^*(G) = \MMO(G)$ 
else $S_{\mathrm{max}}^*(G) = \MMO(G)+1$.
\end{enumerate}
\end{minipage}
}
\end{center}
\end{figure}
 
\begin{remark}
Using ideas from the proof of Theorem \ref{MMO2.thm}, the above algorithm can be modified 
in a straightforward way to construct a scheme attaining the
value $S_{\mathrm{max}}^*(G)$. 
\end{remark}

\section{Discussion and Summary}
\label{discuss.sec}

We have studied a general type of key redistribution scheme based on hashing
secret values along with users' IDs. We gave necessary and sufficient conditions for
schemes of this type to be secure, and we studied the problem of minimizing users' storage
in these schemes.

Choi, Acharya and Gouda \cite{CAG13} also studied key predistibution in a similar setting.
They only considered the situation where the communication graph is a complete graph and they 
gave a construction that is basically equivalent to the one found in \cite{LS05}. They also considered
lower bounds on the total storage of schemes of this type. Their model is very similar to ours,
but they do not assume (as we did) that the key derivation function is a random oracle. In 
\cite[Theorem 6]{CAG13}, it is stated that the total storage required by a secure scheme is
at least $n(n-1)/2$. The proof involves analyzing the implications of an equation of the form
\begin{equation}
\label{CAG.eq}
F(\iv,\ku) = F(\iu,\kv),
\end{equation}
where $F$ is a public key derivation
function, $\iu$ and $\iv$ are public values, and $\ku$ and $\kv$ are secret values known to 
nodes $u$ and $v$ respectively.

The  equation (\ref{CAG.eq}) ensures that nodes $u$ and $v$ will
compute the same pairwise key. Moreover, it is assumed $\ku$ and $\kv$ are 
each used for the computation of at least one other key, and then a contradiction is derived.
It is observed in \cite{LS05} that, under these circumstances, it should be infeasible
to compute $\ku$ given $\iv$ and $F(\iu,\kv)$, and it should also be infeasible
to compute $\kv$ given $\iu$ and $F(\iv,\ku)$.
It is then claimed that this means that values $ku$ and $kv$ satisfying (\ref{CAG.eq}) 
cannot be computed when the scheme is set up. However, this last assertion does not seem to 
consider the possibility that the entity that sets up the scheme can compute these values, even though
the nodes $u$ and $v$ might not be able to compute $\kv$ or $\ku$, respectively.

The Blom scheme \cite{Bl83} illustrates how this can happen.
For simplicity, we consider a Blom scheme secure against individual nodes.
Such a scheme is constructed by a trusted authority (TA) first choosing a symmetric polynomial of the form
$g(x_1,x_2) = a+ b(x_1+x_2) + cx_1x_2$.  A  node $u$ is given the polynomial $f_u(z) = g(iu,z)$
and $v$ is given the polynomial $f_v(z) = g(iv, z)$.
The pairwise key for nodes $u$ and $v$ is $g(\iu,\iv) = g(\iv,\iu)$. Node $u$ computes this key as
$f_u(\iv)$ and node $v$ computes $f_v(\iu)$. Note that node $u$ cannot compute $f_v$ and node 
$v$ cannot compute $f_u$, but the TA who sets up the scheme knows the relationship between
these polynomials.

In the setting we studied, where the key derivation function is a random oracle, the equation
(\ref{CAG.eq}) will not hold, so the above discussion does not apply. We were therefore able to prove
stronger lower bounds on the total storage, as we presented in Section \ref{storage.sec}.

\end{document}